\documentclass[twocolumn,showpacs,preprintnumbers,amssymb, superscriptaddress]{revtex4}
\usepackage{pstricks}%
\usepackage{graphicx}
\usepackage{dcolumn}
\usepackage{bm}
\usepackage{dsfont}
\usepackage{amsmath}
\usepackage{pifont}
\usepackage{psfrag}
\usepackage[latin2]{inputenc}
\usepackage{bbm}

\renewcommand\Re{\text{Re\,}}

\newcommand\Tr{\text{Tr\,}}

\newcommand{\BE}{\begin{equation}}
\newcommand{\EE}{\end{equation}}

\newcommand{\skipc}[2]{}

\newcommand{\fig}[1]{Fig.~\ref{#1}}
\newcommand{\eq}[1]{Eq.~(\ref{#1})}
\newcommand{\Sec}[1]{Sec.~\ref{#1}}
\newcommand{\lem}[1]{Lemma~\ref{#1}}
\newcommand{\theo}[1]{Theorem~\ref{#1}}

\newcommand{\I}{\ensuremath{{\mkern1mu\mathrm{i}\mkern1mu}}}
\newcommand{\E}{\ensuremath{{\mkern1mu\mathrm{e}\mkern1mu}}}

\newcommand{\sep}{\ensuremath{{\mkern1mu\mathrm{sep}\mkern1mu}}}
\newcommand{\ps}{\ensuremath{{\mkern1mu\mathrm{ps}\mkern1mu}}}
\newcommand{\bs}{\ensuremath{{\mkern1mu\mathrm{bs}\mkern1mu}}}
\newcommand{\nc}{\ensuremath{{\mkern1mu\mathrm{nc}\mkern1mu}}}
\newcommand{\PPT}{\ensuremath{{\mkern1mu\mathrm{PPT}\mkern1mu}}}

\newtheorem{theorem}{Theorem}[section]
\newtheorem{lemma}[theorem]{Lemma}

\newcommand{\qed}{\nobreak \ifvmode \relax \else
      \ifdim\lastskip<1.5em \hskip-\lastskip
      \hskip1.5em plus0em minus0.5em \fi \nobreak
     $\square$\fi}

\newenvironment{proof}[1][Proof]{\begin{trivlist}
\item[\hskip \labelsep {\bfseries #1}]}{\end{trivlist}}

\newenvironment{step}[1][-]{\begin{trivlist}
\item[\hskip \labelsep {\bfseries #1}]}{\end{trivlist}}

\renewcommand{\vr}{\varrho}
\newcommand{\be}{\begin{equation}}
\newcommand{\ee}{\end{equation}}
\newcommand{\eea}{\end{eqnarray}}
\newcommand{\bea}{\begin{eqnarray}}
\newcommand{\mean}[1]{\ensuremath{\langle{#1}\rangle}}

\newcommand{\ket}[1]{\ensuremath{|#1\rangle}}
\newcommand{\bra}[1]{\ensuremath{\langle#1|}}
\newcommand{\braket}[2]{\ensuremath{\langle #1|#2\rangle}}
\newcommand{\ketbra}[1]{\ensuremath{| #1 \rangle \!\langle #1 |}}

\begin{document}

\title{A unified approach to entanglement criteria using the 
Cauchy-Schwarz and H\"older inequalities}

\author{Sabine W\"olk}
\affiliation{Naturwissenschaftlich-Technische Fakult\"at, Universit\"at Siegen, Walter-Flex-Str.~3, 57068 Siegen, Germany}
\author{Marcus Huber}
\affiliation{Departament de F\'isica, Universitat Aut\`onoma de Barcelona, 08193 Bellaterra, Spain}
\affiliation{ICFO-Institut de Ci\`encies Fot\`oniques, Mediterranean Technology Park, 08860 Castelldefels (Barcelona), Spain}
\author{Otfried G\"uhne}
\affiliation{Naturwissenschaftlich-Technische Fakult\"at, Universit\"at Siegen, Walter-Flex-Str.~3, 57068 Siegen, Germany}

\date{\today}

\begin{abstract}
We present unified approach to different recent entanglement 
criteria. Although they were developed in different ways, 
we show that they are all applications of a more general 
principle given by the Cauchy-Schwarz inequality. We explain 
this general principle and show how to derive with it not only 
already known but also new entanglement criteria. We systematically 
investigate its potential and limits to detect bipartite and multipartite 
entanglement.
\end{abstract}

\pacs{03.67.-a, 03.65.Ud}

\maketitle

\section{Introduction}

The phenomenon of entanglement is of fundamental interest since it is a
main difference between the classical and the quantum world. Furthermore, 
it is believed to be the central resource for quantum computing and protecting 
quantum communication. The existing entanglement criteria solve the problem of characterizing entanglement only for certain classes of states and, in addition, 
some of them are very resource intensive if applied experimentally. For example,
to apply the famous positive-partial-transpose (PPT) criterion \cite{Peres1996,Horodecki1996} experimentally, a full quantum state tomography 
is necessary in practice. Some other entanglement criteria are formulated 
as inequalities for mean values of observables \cite{Shchukin2005,Hillery2006,Guehne2010,Hillery2010,Duer2000,Huber2010}. 
This sort of entanglement criteria are especially useful to detect entanglement 
in experiments since complete knowledge about the quantum state is not necessary. 

When going from bipartite entanglement to multipartite entanglement, the 
detection and characterization of entanglement become even more complicated. 
First of all, there exist different degrees of entanglement. That is, an 
$N$-partite entangled state $\varrho$ may be a convex combination of pure 
entangled states with maximally $k$ entangled parties. If at least one 
$N$-partite entangled pure state is necessary to form $\varrho$, we call the state 
genuine multipartite entangled. Since the representation of a mixed state 
by a convex sum of pure states is not unique, the entanglement 
characterization of multipartite states is more than the combination 
of bipartite entanglement criteria. For example, there exist states 
which are entangled under every bipartite split but are not genuine 
multipartite entangled. On the other hand, there exist states, which 
are separable under every possible bipartite split, but not fully 
separable.  

Consequently, there are many different criteria and there are many 
different ways to  develop them. To give an example, Hillery and 
Zubairy developed in Ref.~\cite{Hillery2006} first a criterion  
based on uncertainties, followed by a generalized entanglement 
criterion solely based on the Cauchy-Schwarz inequality and 
the properties of separable states. Whereas this approach used 
only a single operator per subsystems, we gain more freedom by  
using two operators per subsystem. In this way, we develop 
in this paper a general principle to create entanglement 
criteria. Many already existing criteria follow immediately 
from this new principle, but also new criteria can be created.

This paper consists of  two main parts: In \Sec{sec:Cauchy} we 
introduce our scheme of developing entanglement criteria with 
the help of the Cauchy-Schwarz inequality. Here, we will first 
concentrate on bipartite entanglement in \Sec{sec:bipartied}. 
We explain  the intimate connection of the PPT criterion
\cite{Peres1996,Horodecki1996} to our criterion and give explicit 
examples of classes of states which can and which can not be detected 
by our criterion. 
Afterwards, in \Sec{sec:multi}, we generalize our scheme to 
multipartite entanglement by using the H\"older inequality
and show that in the multipartite case our 
criterion can detect entangled states which cannot be detected 
by the PPT criterion. In \Sec{sec:other_criteria} we discuss 
several already existing criteria and demonstrate how they can 
be derived within our scheme. In this way, a connection between 
the different criteria are pointed out and the advantages and 
disadvantages of the different criteria become visible.

\section{Entanglement criteria from the Cauchy-Schwarz inequality}
\label{sec:Cauchy}
First, recall that a two-particle state $\vr$ is separable, if it can be
written as a mixture of product states
\BE
\vr= \sum_{k} p_k \vr_k^{A} \otimes \vr_k^{B},
\EE
where the $p_k$ form a probability distribution. We aim at deriving entanglement
criteria in the form of inequalities which hold for separable states, but which
can be violated by entangled states. Our main tool are two simple facts: 
\\
(i) we use the property of product states $\vr=\vr^A \otimes \vr^B$ that
\BE
\langle  A   B\rangle_{\rm ps}= \mean{A}_{\vr^A} \mean{B}_{\vr^B}
\EE
for operators $A$ and $B$ acting on Alice's and Bobs subsystem, 
respectively. Here and in the following, $\mean{\dots}_{\rm ps}$ 
denotes that the expectation value is taken for a product state. 
Similarly $\mean{\dots}_{\rm sep}$ denotes an expectation value 
for a separable state.
\\
(ii) The Cauchy Schwarz (CS) inequality \cite{Cauchy}
\BE
|\langle x,y\rangle |^2 \leq \langle x,x\rangle \langle y,y\rangle
\EE
where $x$ and $y$ are two vectors and $\langle \cdot,\cdot\rangle$ defines an inner product.

\subsection{Bipartite entanglement\label{sec:bipartied}}

We concentrate on  expectation values of a bipartite system which 
can be written as $\langle   A_1   A_2  B_1   B_2\rangle$ where the 
operators $  A_i$  and $  B_j$ acting on  Alice's and Bob's subsystem, 
respectively.  For a pure state, these expectation values can be 
interpreted as the inner product of the two vectors  
$\ket{\psi_1} \equiv   A_1^\dagger   B_1^\dagger\ket{\psi}$ 
and $\ket{\psi_2}\equiv   A_2   B_2\ket{\psi}$. Since the scalar product 
is bilinear, all fragmentation of a scaler product into a bra- and a 
ket-vector can be described by two operators per subsystems. More than 
two are not necessary, since they can be always combined to a single 
operator acting on the bra- and one acting on the ket-vector. The 
expectation value $\langle   A_1   A_2  B_1   B_2\rangle$ can be 
used to detect entanglement by using the following theorem:


\begin{theorem}\label{theo:cauchy}
 The inequality 
 \BE
 | \langle   A_1   A_2  B_1   B_2\rangle_\sep|^2\leq  \langle   A_1  A_1^\dagger   B_2^\dagger  B_2\rangle_\sep\langle   A_2^\dagger  A_2 B_1  B_1^\dagger\rangle_\sep. \label{eq:Cauchy2}
 \EE
is valid for separable states and can only be violated by entangled states.
\end{theorem}

\begin{proof}
 For product states $\ket{\phi}=\ket{a}\ket{b}=\ket{a,b}$, we can write
 the expectation value as the product of the expectation values of the 
 subsystems. By applying the CS inequality to every single 
 subsystem we get
\begin{eqnarray}
| \bra{a,b}   A_1   A_2  B_1   B_2\ket{a,b}|^2&=&| \bra{a}   A_1   A_2 \ket{a}|^2|\bra{b} B_1   B_2\ket{b}|^2 \nonumber\\
&\leq&  \bra{a}   A_1  A_1^\dagger \ket{a}\bra{a}   A_2^\dagger  A_2 \ket{a} \nonumber \\ &&\times \bra{b} B_1  B_1^\dagger\ket{b} \bra{b}   B_2^\dagger  B_2\ket{b} \label{eq:Cauchy1}
\end{eqnarray}
which  is equal to
\BE
| \langle   A_1   A_2  B_1   B_2\rangle_\ps|^2\leq  \langle   A_1  A_1^\dagger   B_2^\dagger  B_2\rangle_\ps\langle   A_2^\dagger  A_2 B_1  B_1^\dagger\rangle_\ps.\label{eq:Cauchy_prod}
\EE
In order to see that this inequality holds also for mixed states, note first that
the root of the left-hand side $|\langle   A_1   A_2  B_1   B_2\rangle_\ps|$ is
convex in the state, while the root of the right-hand side is of the type $\sqrt{f(\vr) g(\vr)}$, where $f$ and $g$ are positive functions. This implies that it is concave \cite{Guehne2010}. Therefore, this inequality is  also valid for mixtures of product states
which proves the correctness of \eq{eq:Cauchy2} for any separable state.

\qed
\end{proof}

With the help of the CS inequality we derive for general states
\BE
| \langle   A_1   A_2  B_1   B_2\rangle|^2\leq  \langle   A_1  A_1^\dagger   B_1  B_1^\dagger\rangle \langle   A_2^\dagger  A_2   B_2^\dagger  B_2\rangle\label{eq:cauchy_ent},
\EE
which provides a tight upper limit for all states.

Depending on the choice of the operators $  A_j$ and $  B_j$ the inequality  
\eq{eq:Cauchy2} may provide a stricter limit than \eq{eq:cauchy_ent}. If yes, 
then there exist states which violates  \eq{eq:Cauchy2}. These states must 
be entangled and therefore \eq{eq:Cauchy2} is able to detect entanglement.

Regarding the application of \theo{theo:cauchy} we note:
(i) The ability of \eq{eq:Cauchy2} to detect entanglement 
depends on the chosen operators $  A_j$ and $  B_j$. For example, 
by a comparison of \eq{eq:cauchy_ent} and \eq{eq:Cauchy2} we 
immediately find that if $  A_1  A_1^\dagger=  A_2^\dagger   A_2$ or $  B_1  B_1^\dagger=  B_2^\dagger   B_2$ \eq{eq:Cauchy2} cannot be violated for any state. 
(ii) The optimal choice of the operators $  A_j$ and $  B_j$ depends in 
general on the state $\varrho$. (iii) Whereas \eq{eq:Cauchy2} is valid 
for general mixed separable states, \eq{eq:Cauchy1} used to develop 
\eq{eq:Cauchy2} is only valid for pure product states. As a consequence, 
great care concerning convexity has to be taken when generalizing 
\eq{eq:Cauchy2} to multipartite systems.

After the development of the entanglement criteria \eq{eq:Cauchy2} 
we will investigate the bipartite case now in more detail. First, 
we discuss the best choice of operators.

   
\begin{theorem}\label{theo:opti}
 The best choice of the operators $  A_j$ are given by
 \BE
   A_1 = \ket{a}\bra{\varphi};\;   A_2 = \ket{\varphi}\bra{\alpha}
 \EE
 with $\ket{a}$, $\ket{\alpha}$ and $\ket{\varphi}$ being  pure 
 states of Alice's subsystem, and an analogous  choice for Bob.
\end{theorem}

\begin{proof}
Any pair of operators $A_1$ and $  A_2$ can be written 
as
\begin{eqnarray}
    A_1 &=& \sum\limits_j x_j\ket{a_j}\bra{\varphi_j}\\
     A_2 &=& \sum\limits_j y_j\ket{\varphi_j}\bra{\alpha_j}
\end{eqnarray}
 with $\lbrace \ket{\varphi_j}\rbrace$ being an orthonormal basis 
 of Alice's subsystem. As a consequence, the entanglement criterion 
 \eq{eq:Cauchy2} turns into
\BE
\Big|\sum\limits_j x_jy_j \Tr[\ket{a_j}_A\bra{ \alpha_j}  B_1  B_2 \varrho]\Big|^2 \leq \sum\limits_j |x_j|^2p_j\sum\limits_k |y_k|^2q_k
\EE
where we defined
 \begin{eqnarray}
 p_j&\equiv& \Tr[\ket{a_j}\bra{a_j}  B_2^\dagger  B_2 \varrho], \nonumber 
 \\
  q_k&\equiv& \Tr[\ket{ \alpha_k}\bra{ \alpha_k}  B_1  B_1^\dagger \varrho].
\end{eqnarray}
The right hand side of this equation is equal to
\BE
  \begin{split}
 RS1=&\sum\limits_j |x_j|^2|y_j|^2 p_jq_j\\
 &+ \sum\limits_j\sum\limits_{k>j}|x_j|^2|y_k|^2p_jq_k+|x_k|^2|y_j|^2 p_k q_j
 \end{split}\label{eq:LS1}
 \EE
 On the other hand, we can estimate the left hand side by
 \BE
 \begin{split}
 &\Big|\sum\limits_j x_jy_j \Tr[\ket{a_j}\bra{ \alpha_j}  B_1  B_2 \varrho]\Big|^2\\&\leq \Big(\sum\limits_j |x_j||y_j| \left|\Tr[\ket{a_j}_A\bra{ \alpha_j}  B_1  B_2 \varrho]\right|\Big)^2
 \end{split}
 \EE
 and use \eq{eq:Cauchy2} for every single term in the summation. This 
 leads to
 \BE
 \Big|\sum\limits_j x_jy_j \Tr[\ket{a_j}\bra{ \alpha_j}  B_1  B_2 \varrho]\Big|^2\leq \Big(\sum\limits_j |x_j||y_j|\sqrt{p_jq_j}\Big)^2
 \EE
 By expanding the right hand side we get
\BE
 \begin{split}
 RS2=&\sum\limits_j |x_j|^2|y_j|^2p_jq_j\\&+\sum\limits_j\sum\limits_{k>j} 2|x_jx_ky_jy_k|\sqrt{p_jp_kq_jq_k} \label{eq:LS2}
  \end{split}
 \EE
By using $r^2+s^2\geq 2 rs$ valid for all real numbers $r,s$ and identifying
\begin{eqnarray}
r&=&|x_j||y_k|\sqrt{p_jq_k}\\
s&=&|x_k||y_j|\sqrt{p_kq_j}
\end{eqnarray}
for all arbitrary but fixed combination of $j$ and $k$ we find by 
comparing \eq{eq:LS1} with \eq{eq:LS2}
\BE
RS2 \leq RS1.
\EE
As a consequence, using several entanglement criteria with operators
\BE
  A_1 = \ket{a_j}\bra{\varphi_j};\;   A_2 = \ket{\varphi_j}\bra{\alpha_j}
\EE
for every $j$ separately leads to a stronger criterion than a  linear combination of these operators.\qed
\end{proof}
Naturally the same holds for Bob's subsystems {which leads to the well known 
entanglement criterion \cite{Guehne2010, Huber2010}
\BE
\left|\bra{\alpha,\beta}\varrho\ket{a,b}_\sep\right|\leq \sqrt{\bra{a,\beta}\varrho\ket{a,\beta}_\sep\bra{\alpha,b}\varrho\ket{\alpha,b}_\sep }\label{eq:optbisep}.
\EE}
   
Now, we turn to the question which states can be detected by \theo{theo:cauchy}. 
Similar to the criterion from Hillery and Zubairy \cite{Hillery2009}, our 
criterion is strongly connected to the positive partial transpose (PPT) 
criterion in the bipartite case \cite{Peres1996,Horodecki1996}.

\begin{theorem}\label{theo:PPT}
  The criterion \theo{theo:cauchy} for detecting bipartite entanglement 
  detects only states with a negative partial transpose (NPT).  
\end{theorem}

\begin{proof}
 By using the partial transpose with respect to system A (written as $T_A$) acting on the state $\varrho$ as well as on the operators $  A_1$ and $  A_2$ we rewrite the expectation value 
 \BE|\langle   A_1   A_2  B_1   B_2\rangle|= \left|\Tr[  A_2^T   A_1^T   B_1   B_2 \varrho^{T_A}]\right|.
 \EE
 If $\varrho$ is PPT, than $\varrho^{T_A}$ is a valid density operator and the Cauchy-Schwarz inequality for expectation values
 \BE
 \left|\Tr[O P \varrho]\right|\leq \sqrt{\Tr[O O^\dagger \varrho]\Tr[P^\dagger P \varrho]}
 \EE
 can be applied. This leads to 
\BE
\begin{split}
 \left|\Tr[  A_2^T   A_1^T   B_1   B_2 \varrho_\PPT^{T_A}]\right| \leq& \sqrt{\Tr[  A_2^T \left(  A_2^T\right)^\dagger   B_1   B_1^\dagger \varrho_\PPT^{T_A}]}\\ &\times \sqrt{\Tr[ \left(  A_1^T\right)^\dagger   A_1^T    B_2^\dagger    B_2\varrho_\PPT^{T_A}]
 }\label{eq:cauchy_TA}
 \end{split}
 \EE
 were we identified $O=A_2^TB_1$ and $P=A_1^TB_2$.
Writing the expectation value again with respect to $\varrho$ we finally arrive at
  \BE
  \begin{split}
  |\langle   A_1   A_2  B_1   B_2\rangle_\PPT|\leq& \sqrt{\langle  A_2^\dagger   A_2   B_1   B_1^\dagger \rangle_\PPT \langle   A_1  A_1^\dagger    B_2^\dagger    B_2 \rangle_\PPT}
  \end{split}
  \EE
  which is equal to \eq{eq:Cauchy2}. As a consequence, PPT states are not 
  able to violate \eq{eq:Cauchy2}.\qed
 \end{proof}
 
  We note that \theo{theo:PPT} is only valid for the bipartite case. In the multipartite case, also PPT-states can be detected, since our criteria not only checks
 bipartite entanglement but real multipartite entanglement as we will show in the next section. However, we will once again stay in the bipartite case  and identify now a few classes of entangled states which can be detected by \eq{eq:Cauchy2}. Furthermore, we will also show how to find the right operators $  A_j$ and $  B_j$ for these cases. To achieve this task, the following lemma will be helpful: 
    
\begin{lemma}\label{lemma:operators}
 Let $\ket{\lambda_j^+}$ be the eigenstates of $\varrho^{T_A}$ corresponding to the
 positive eigenvalues and $\ket{\lambda_j^-}$ the ones corresponding to negative eigenvalues. Furthermore, we assume that there exist  states $\lbrace\ket{a_k}\rbrace$ and $\lbrace\ket{b_k}\rbrace$ of Alice's and Bobs subsystems, respectively, with $k \in \lbrace 1,2\rbrace$ such that
 \begin{eqnarray}
  \braket{a_2,b_1}{\lambda_j^+}&=&  c^+ \braket{a_1,b_2}{\lambda_j^+}\\
  \braket{a_2,b_1}{\lambda_j^-}&=&  c^-\braket{a_1,b_2}{\lambda_j^-}
 \end{eqnarray}
with   $c^\pm$ independent of $ j$ and $c^+c^-<0$. By choosing the operators
\begin{eqnarray}
      A_1  = \ket{a_1^\ast}\bra{\alpha}&& B_1= \ket{b_1}\bra{\beta}\\
      A_2 = \ket{\alpha}\bra{a_2^\ast}&&B_2= \ket{\beta}\bra{b_2}  
\end{eqnarray}
with $\ket{\alpha},\ket{\beta}$ being arbitrary states of Alice's and Bobs subsystems, respectively, \eq{eq:Cauchy2} is violated  by the state $\varrho$.
\end{lemma}

\begin{proof} Although the state under consideration is not PPT, we can
use some calculations
from the proof of Theorem \ref{theo:PPT}.
 By defining $p_j= \braket{a_1,b_2}{\lambda_j^+}$ and $q_j= \braket{a_1,b_2}{\lambda_j^-}$ we obtain for the left hand side of \eq{eq:cauchy_TA}
 \begin{eqnarray}\label{eq:LS}
 LS&\equiv&\left|\bra{a_1,b_2}\varrho^{T_A}\ket{a_2,b_1}\right|\\
 &=&|c^+|\sum_j \lambda_j^+ |p_j|^2 + |c^-|\sum_j |\lambda_j^-| |q_j|^2.\label{eq:LS}
 \end{eqnarray}
 The right hand side becomes
\begin{eqnarray}
 RS&\equiv&\sqrt{\bra{a_2,b_1}\varrho\ket{a_2,b_1}\bra{a_1,b_2}\varrho\ket{a_1,b_2}}
 \\&=& \sqrt{|c^+|^2\sum_j \lambda_j^+ |p_j|^2-|c^-|^2\sum_j|\lambda_j^-||q_j|^2}\nonumber \\ &\times& \sqrt{\sum_j \lambda_j^+ |p_j|^2-\sum_j|\lambda_j^-||q_j|^2}.
 \end{eqnarray}
 By expanding the product we arrive at
 \BE
 \begin{split}
 RS=&\Big[|c^+|^2 \Big(\sum_j \lambda_j^+ |p_j|^2\Big)^2 +|c^-|^2 \Big(\sum_j |\lambda_j^-| |q_j|^2\Big)^2 \\&-\left(|c^+|^2+|c^-|^2\right) \Big(\sum_j \lambda_j^+ |p_j|^2\Big)\Big(\sum_j |\lambda_j^-| |q_j|^2\Big)\Big]^{1/2}.
 \end{split}
 \EE
 With the help of $-\left(|c^+|^2+|c^-|^2\right)\leq -2|c^+||c^-|$ we are able to estimate the right side by
 \BE\label{eq:RS}
 RS\leq \Big||c^+|\sum_j \lambda_j |p_j|^2 - |c^-|\sum_j |\lambda_j^-| |q_j|^2\Big|.
 \EE
 A comparison of  \eq{eq:LS} and  \eq{eq:RS} shows that the left hand side of \eq{eq:cauchy_TA} is the summation of two positive numbers whereas the right side is smaller or equal than the absolute value of the difference of the same two positive numbers. As a consequence, we have $LS > RS$ and a violation of \eq{eq:cauchy_TA} 
 which implies a violation of criterion \theo{theo:cauchy}.\qed

\end{proof}

 
 Now, we are able to identify  certain classes of entangled states
 that can be detected with our criterion. One of these classes are 
 bipartite qubits states:

 \begin{theorem}
  Every entangled two-qubit state $\varrho $ can be detected
  with \theo{theo:cauchy}
 \end{theorem}
 
 \begin{proof}
 Every entangled two-qubit state is an NPT-state. All eigenstates $\ket{\lambda_j^-}$  corresponding to a negative eigenvalue $\lambda_j^-$ of the partial transpose of $\varrho$   form an entangled subspace \cite{Johnston2013,Rana2013}. That means it is not possible to construct a product state by a superposition of $\ket{\lambda_j^-}$. As a consequence, for two-qubit states  only one single negative eigenvalue can exist \cite{Sanpera1998}. Its eigenstate is given in its Schmidt basis by
  \BE
  \ket{\lambda^-}\equiv s_0\ket{00}+s_1\ket{11}.
  \EE
  with $s_j>0$.   Since a two-qubit state is separable iff the determinant of its partial transpose is nonnegative \cite{Horodecki2008}, $\varrho^T_A$ has three strictly positive eigenvalues $\lambda^+_j>0$. Their corresponding eigenstates are given by 
  \BE
  \ket{\lambda_k^+}\equiv  r_{k} \left(s_1\ket{00}-s_0\ket{11}\right)+\gamma_k\ket{01}+\delta_k\ket{10}
  \EE
  with arbitrary coefficients $r_k$, $\gamma_{k}$ and $\delta_{k}$ and at least one $k$ for which $r_k\neq 0$.
 By choosing the  states $\ket{a_2,b_1}=\ket{00}$ and $\ket{a_2,b_1}=\ket{11}$ we find
  \BE
   \braket{00}{\lambda_k^+}=r_{k} s_1 \,,\, \braket{11}{\lambda_k^+}=-r_{k} s_0.
  \EE
  Therefore the constant $c^+$  defined in \lem{lemma:operators} does exist for these states and is given by  $c^+=-s_1/s_0$. Similar we find $c^-=s_0/s_1$ and as a consequence \lem{lemma:operators} can be applied. \qed
 \end{proof}
  
  We want to stress out the fact that $\varrho$ is entangled iff the determinant of 
  $\varrho^T_A$ is negative is not necessary for our proof. Moreover, the existence 
  of at least one $r_k\neq 0$ can be proven from our previous considerations:
  
  Assume that there were an NPT two-qubit state with a vanishing eigenvalue, such that 
  $r_{k}=0$ for all $k$. Since the trace is preserved under partial transposition,
  there must exist at least one positive eigenvalue and therefore  there is some $\gamma_{k}\neq 0$ or $\delta_k\neq 0$. Let us  assume $\gamma_{k_0}\neq 0$. It follows
  that we can choose the product states $\ket{a_2,b_1}=\ket{00}+\ket{01}$ and $\ket{a_1,b_2}=\ket{00}-\ket{01}$ and find
 \BE
   \braket{a_2b_1}{\lambda_k^+}=\gamma_{k}  \,,\, \braket{a_1,b_2}{\lambda_k^+}=-\gamma_{k} .
  \EE
  Therefore, the constant $c^+$ is given by $c^+=-1$ and similar $c^-=1$. By using \lem{lemma:operators} we get a violation of \eq{eq:cauchy_TA}. However, in this 
  case we would have $  A_1  A_1^\dagger=\ket{0}\bra{0}=  A_2^\dagger  A_2$ and for 
  these operators \eq{eq:Cauchy2} can never be violated. As a consequence if 
  $\ket{\lambda_j^-}\equiv \alpha \ket{00}+\beta \ket{11}$ is an eigenvector 
  of $\varrho^{T_A}$ with a negative eigenvalue then there must be an eigenvector $\ket{\lambda_j^+}$ with
  \BE
  \bra{\lambda_j^+}\left(s_1\ket{00}-s_0\ket{11}\right)\neq 0.
  \EE

  
  Also in higher dimensions there exist classes of entangled states 
  which can be detected with our criterion. One of these classes are 
  entangled states mixed with with noise:

  \begin{theorem}\label{theorem:wn}
   Every  NPT-state of the form
   \BE
   \varrho_\text{wn}=p \ket{\psi_\text{ent}}\bra{\psi_\text{ent}}+\frac{1-p}{D}\mathbbm{1}_D,
   \EE
   with $\mathbbm{1}_D$ being the identity matrix of dimension $D$ of $\varrho$ 
   and $\ket{\psi_\text{ent}}$ being a pure entangled state, can be detected by \eq{eq:Cauchy2}.
  \end{theorem}
  
  To prove this theorem, we use the following fact, that can be proved by direct 
  calculation:
    \begin{lemma}\label{lemma:wn}
   The partially transposed state  $\varrho_\text{wn}^{T_A}$ 
   with  $\ket{\psi_\text{ent}}$  given in its Schmidt basis by
   $\ket{\psi_\text{ent}}=\sum_{j} s_j \ket{jj}$
   has the following eigenvalues and eigenstates:
   \begin{eqnarray}
    \ket{\lambda_{j,j}}=\ket{j,j} &, &\lambda_{j,j}= \frac{1-p}{D}+p s_j^2\\
     \ket{\lambda_{j,k}}=\frac{\ket{j,k}\pm \ket{k,j}}{\sqrt{2}} &,& \lambda_{j,k}= \frac{1-p}{D}\pm p s_j s_k
   \end{eqnarray}
  \end{lemma}
  
 Now we are able to prove the theorem:
   
   \begin{proof}
   From the previous Lemma one finds that an eigenstate corresponding to negative 
   eigenvalues is of the form
    \BE
    \ket{\lambda_-}=\frac{\ket{j_0,k_0}- \ket{k_0,j_0}}{\sqrt{2}}
    \EE
    and that  
    \BE
    \ket{\lambda_+}=\frac{\ket{j_0,k_0}+ \ket{k_0,j_0}}{\sqrt{2}}
    \EE
   is an  eigenstate corresponding to a positive eigenvalue. Furthermore, $\braket{j_0,k_0}{\lambda_n}=\braket{k_0,j_0}{\lambda_n}=0$ for all other eigenstates $\ket{\lambda_n}$. 
    By choosing $\ket{a_2,b_1}=\ket{j_0,k_0}$ and $\ket{a_1,b_2}=\ket{k_0,j_0}$ we find
   \begin{eqnarray}
    \braket{a_2,b_1}{\lambda_-}&=&-\braket{a_1,b_2}{\lambda_-}\\
    \braket{a_2,b_1}{\lambda_+}&=&+\braket{a_1,b_2}{\lambda_+}\\
    \braket{a_2,b_1}{\lambda_n}&=&\pm\braket{a_1,b_2}{\lambda_n}=0.
    \end{eqnarray}
    As a consequence the constants $c^+$ and $c^-$ of \lem{lemma:operators} 
    do exist and therefore the entanglement of $\varrho_\text{wn}^\text{PT}$ 
    can be detect with the help of \theo{theo:cauchy}.\qed
   \end{proof}

It is important to point out a consequence of the optimality of choosing $A_{1/2}$ and $B_{1/2}$ such that the criterion takes the form of eq.~(\ref{eq:optbisep}). There are only up to four independent vectors (in the previous notation $|a\rangle,|b\rangle,|\alpha\rangle,|\beta\rangle$) appearing in the inequality. If we denote $|\alpha\rangle=\lambda|a\rangle+\sqrt{1-|\lambda|^2}|a^\perp\rangle$ with $\lambda=\langle a|\alpha\rangle$ it becomes obvious that the criterion is invariant under prior projection into the qubit subspace on Alice's side given by $\mathbbm{1}_2=|a\rangle\langle a|+|a^\perp\rangle\langle a^\perp|$. As the same holds for Bob's side we can conclude that the bipartite version of this theorem is actually equivalent to detecting entanglement in a $2\times 2$ dimensional subspace. From Ref.~\cite{horodecki98a} it follows that any violation of the criterion implies one-copy distillability, which on the other hand implies the impossibility of detecting some NPT states, such as the well known 
Werner state \cite{Werner1989}. In fact it also implies that we can strictly improve the detection strength of the criterion by applying it on multiple copies of the state, as there are known examples of states that are two-, but not one-copy distillable. However some NPT states are most likely beyond the reach of our criterion as e.g. the Werner state is conjectured to be NPT bound entangled, i.e. even up to infinitely many copies do not have an entangled $2\times 2$ subspace for a specific region of parameters.\\
A natural extension of our theorem that would close this (small) gap for bipartite systems is an extension to $D\times D$-dimensional subspaces, which could be written as
\begin{align}
\text{Det}\begin{pmatrix} \rho_{i_1i_1i_1i_1} & \rho_{i_1i_2i_2i_1} & \cdots & \rho_{i_1i_Di_Di_1} \\\rho_{i_2i_1i_1i_2} & \rho_{i_2i_2i_1i_2} & \cdots & \vdots \\\vdots & \vdots & \ddots &\vdots\\\rho_{i_Di_1i_1i_D}&\cdots&\cdots&\rho_{i_Di_Di_Di_D} \end{pmatrix}\geq0\,,
\end{align}
with the notation $\langle a b|\rho_{PPT}|\alpha\beta\rangle=\rho_{ab\alpha\beta}$. This criterion is a natural extension of our main theorem to higher dimensions and in this form is in principle capable of detecting all NPT states, not only one-copy distillable ones. The caveat here is however, that due to the non-convex structure we can neither make use of the Cauchy-Schwarz inequality, nor find an analogue extension to the more interesting case of many particles. So instead we focus on our main theorem, which allows for a straightforward generalization to multipartite systems.

\subsection{Multiparticle entanglement}
\label{sec:multi}

Pure product states and mixed product states of the 
form $\varrho=\bigotimes_{j=1}^N \varrho_j$ with $\varrho_j$ 
do not contain any correlations (nc). The generalization 
of \eq{eq:Cauchy1} to multipartite systems for these 
states is straightforward:
\BE
|\langle \bigotimes\limits_{k=1}^{N}  P_k Q_k\rangle_\nc| \leq \prod\limits_{k=1}^N \sqrt{\langle   P_k   P_k^\dagger\rangle_\nc \langle  Q_k^\dagger   Q_k\rangle_\nc}\label{multi_cauchy_prod1}
\EE
with operators $  P_k$ and $ Q_k$ acting on subsystem $k$.
Again, we can combine expectation values of different subsystems and arrive
at a generalization of \theo{theo:cauchy}: 
\BE
|\langle \bigotimes\limits_{k=1}^{N}  P_k Q_k\rangle_\nc| \leq \prod\limits_{k=1}^N \sqrt{\langle   P_k   P_k^\dagger  Q_{k+1}^\dagger  Q_{k+1} \rangle}_\nc.\label{multi_cauchy_prod2}
\EE
where we defined $ Q_{N+1}=Q_{1}$. Also other combination of expectation values are possible. However, these inequalities are not convex and therefore in general not valid for  mixed separable states.

To generalize entanglement criteria of the form of \eq{multi_cauchy_prod2} 
to mixed fully separable states, we have to use the H\"older  
inequality. We will explain the generalization using the example of 
\eq{multi_cauchy_prod2} for a tripartite state with the operators 
$A_j$, $  B_j$ and $  C_j$ acting on Alice's, Bob's and Charlie's subsystem, 
respectively. 

\subsubsection{Scheme to develop entanglement criteria \label{sec:scheme}}
Performing the following steps will result in an inequality to detect 
entanglement. We assume that we start with an fully separable state
\be
\vr = \sum_j p_j \ketbra{\psi^\ps_j}
\ee
where the $\ket{\psi^\ps_j}$ are pure multiparticle product states as above.

\begin{step}
Write the expectation value $\langle   A_1  A_2  B_1  B_2  C_1  C_2\rangle$ of 
the fully seperable state as a convex combination of expectation values for 
pure product states $\ket{\psi^\ps_j}$. 
\BE
\langle   A_1  A_2  B_1  B_2  C_1  C_2\rangle_\text{sep}= \sum\limits_j p_j\langle   A_1  A_2  B_1  B_2  C_1  C_2\rangle_{\ps,j}  
\EE
with $\sum_j p_j=1$ and $\langle \cdots \rangle_{\ps,j}$ denoting the expectation value of state $\ket{\psi^\ps_j}$.
\end{step}

\begin{step}
 Write every expectation value $\langle \cdots \rangle_{\ps,j}$ as expectation values of single subsystems and apply the Cauchy Schwarz inequality to each of them
 \BE
|\langle   A_1  A_2  B_1  B_2  C_1  C_2\rangle_\sep|\leq \sum\limits_j p_j\sqrt{\langle   A_1  A_1^\dagger\rangle_j \langle  A_2^\dagger   A_2\rangle_j \cdots  }.
\EE
Since operators of the form $  O=  D   D^\dagger$ are positive operators, it is also possible to increase the number of expectation values by using $\langle   O\rangle = \sqrt[n]{\langle   O\rangle^n}$ .
\end{step}

\begin{step}
 Combine arbitrary expectation values of different subsystems, for example
\BE
\begin{split}
&|\langle   A_1  A_2  B_1  B_2  C_1  C_2\rangle_\sep| \\
\leq& \sum\limits_j p_j \sqrt{\langle   A_1  A_1^\dagger   B_2^\dagger   B_2\rangle_{j} \langle   B_1  B_1^\dagger   C_2^\dagger   C_2\rangle_{j}\langle   C_1  C_1^\dagger   A_2^\dagger   A_2\rangle_{j}}
\end{split}
\EE
It is also possible to combine more than two subsystems in a single expectation value.
\end{step}

\begin{step}
 Use the generalized H\"older inequality \cite{Hardy34}
\BE
\sum\limits_j p_j x_j y_j\leq \Big(\sum\limits_j p_j x_j^{1/r}\Big)^r
\Big(\sum\limits_j p_j y_j^{1/s}\Big)^s\label{eq:hoelder}
\EE 
with $r+s=1$ and $p_j\geq 0$ to separate the summation of each 
expectation value. Here, the H\"older inequality may be used several
times. In our example, two applications of the inequality lead to:
\BE
\begin{split}
 |\langle   A_1  A_2  B_1  B_2  C_1  C_2\rangle_\sep| \leq& \Big(\sum\limits_j p_j \langle   A_1  A_1^\dagger   B_2^\dagger   B_2\rangle_{\ps,j}^{3/2}\Big)^{1/3} \\ &\times \Big(\sum\limits_j p_j \langle   B_1  B_1^\dagger   C_2^\dagger   C_2\rangle_{\ps,j}^{3/2}\Big)^{1/3} \\ &\times
 \Big(\sum\limits_j p_j \langle   C_1  C_1^\dagger   A_2^\dagger   A_2\rangle_{\ps,j}^{3/2}\Big)^{1/3}.
\end{split}
 \EE
\end{step}

\begin{step}
Use the inequality 
\BE
\langle   O\rangle^x\leq \langle   O^x\rangle \label{eq:posop}
\EE
for positive operators $O$ and $x\geq 1$ to make the expectation values 
appear only linearly. Then, everything can be written in terms of $\vr$
again, and the $p_j$ disappear. In our example we finally arrive at
\BE
\begin{split}
 & |\langle   A_1  A_2  B_1  B_2  C_1  C_2\rangle_\sep|^2\\ \leq&\sqrt[3]{\langle  
 \big(  A_1  A_1^\dagger   B_2^\dagger   B_2\big) ^{3/2}\rangle_\sep\langle  \big(  B_1  B_1^\dagger   C_2^\dagger   C_2\big) ^{3/2}\rangle_\sep}\\
 & \times \sqrt[3]{\langle  \big(  C_1  C_1^\dagger   A_2^\dagger   A_2\big) ^{3/2}\rangle_\sep},
  \end{split}
\EE
which is valid for all separable mixed states but can be violated by entangled states.
\end{step}

\subsubsection{Application\label{sec:application}}
In the same way, the criterion
\BE
\begin{split}
  &|\langle   A_1  A_2  B_1  B_2  C_1  C_2\rangle_\sep|\\
  \leq&\sqrt[4]{\langle   A_1  A_1^\dagger   B_1  B_1^\dagger   C_2^\dagger   C_2\rangle_\sep\langle   A_1  A_1^\dagger   B_2^\dagger  B_2   C_1   C_1^\dagger\rangle_\sep} \\ &\times \sqrt[4]{\langle   A_2^\dagger  A_2   B_1  B_1^\dagger   C_1   C_1^\dagger\rangle_\sep\langle   A_2^\dagger   A_2  B_2^\dagger  B_2   C_2^\dagger   C_2\rangle_\sep} \label{eq:Cauchy4}
  \end{split}
\EE
can be shown. This criterion checks all possible bipartitions 
simultaneously and cannot be seen as a combination of criteria 
for bipartite entanglement. Therefore it is possible to detect 
entanglement of states, which are biseperable under every bipartition 
but not fully separable, for example the bound entangled states from
Ref.~\cite{Acin2001}
\BE
  \varrho_{abc} =\frac{1}{n}\left(\begin{array}{cccccccc} 
  1&0&0&0&0&0&0&1\\0&a&0&0&0&0&0&0\\0&0&b&0&0&0&0&0\\0&0&0&1/c&0&0&0&0\\
  0&0&0&0&c&0&0&0\\0&0&0&0&0&1/b&0&0\\0&0&0&0&0&0&1/a&0\\1&0&0&0&0&0&0&1
  \end{array}\right)\label{def:rhoabc}
\EE 
for $abc\neq 1$ with $n=2+a+b+c+1/a+1/b+1/c$. These states are known to 
be separable for any bipartition, but not fully separable. By choosing
\begin{eqnarray}
  A_1=  B_1=  C_1& = &\ket{1}\bra{0}\label{def:op1} \nonumber\\
  A_2=  B_2=  C_2& = &\ket{0}\bra{0}\label{def:op2}.
\end{eqnarray}
we detect entanglement if $1< (abc)^{-1/4}$. A second criterion 
can be gained by combining the expectation values of single systems 
in another way. This second criterion detect entanglement if 
$1< (abc)^{+1/4}$. As a consequence, with the help of the CS inequality 
it can be proven that the state $\varrho_{abc} $ is entangled for 
$abc\neq 1$. A similar result was obtained in Ref.\cite{Guehne2010}. 
However, as we will show in the next section, our criterion is in 
general stronger.

Since the operator $\ket{1}\bra{0}$ is a non-hermitian, its expectation 
value cannot be directly measured. However, since 
$\ket{1}\bra{0}=\sigma_x-\I\sigma_y$ the expectation values  
can be determined by measuring the expectation values of the 
Pauli matrices in $x$ and $y$ direction.

A second interesting example is the bound entangled state from
Ref.~\cite{Kay2011},
\BE \label{def:rhoalpha}
\varrho_\alpha=\frac{1}{8+8\alpha}\left(\begin{tabular}{cccccccc}
	                 4+$\alpha$ &0&0&0&0&0&0&2 \\
	                 0&$\alpha$&0&0&0&0&2&0\\
	                 0&0&$\alpha$&0&0&-2&0&0\\
	                 0&0&0&$\alpha$&2&0&0&0\\
	                 0&0&0&2&$\alpha$&0&0&0\\
	                 0&0&-2&0&0&$\alpha$&0&0\\
	                 0&2&0&0&0&0&$\alpha$&0\\
	                 2 &0&0&0&0&0&0&4+$\alpha$ \\
	                \end{tabular}
	                \right)
\EE
This is a valid state for $2\leq \alpha$. This state is entangled 
(but separable for any bipartition) for $2\leq \alpha \leq 2\sqrt{2}$ 
and separable for $2\sqrt{2}\leq \alpha$ 
\cite{Kay2011, 2011PhLA..375..406G}

We use our criterion \eq{eq:Cauchy4} with operators similar 
to \theo{theo:opti}. Therefore, we define the quantity
\BE
\begin{split} E&\equiv |{\bra{\alpha_1,\beta_1,\gamma_1}}\varrho {\ket{\alpha_2,\beta_2,\gamma_2}}|
	 \\ & - \left(\sqrt[4]{\bra{{\alpha_1,\beta_1},{\gamma_2}}\varrho \ket{{\alpha_1,\beta_1},{\gamma_2}}\bra{{\alpha_1},{\beta_2},{\gamma_1}}\varrho \ket{{\alpha_1},{\beta_2},{\gamma_1}}}\right.\\
	 &\times \left.\sqrt[4]{\bra{{\alpha_2},{\beta_1,\gamma_1}}\varrho \ket{{\alpha_2},{\beta_1,\gamma_1}}{\bra{\alpha_2,\beta_2,\gamma_2}}\varrho {\ket{\alpha_2,\beta_2,\gamma_2}}}\right)\end{split}
\EE
which indicates entanglement if $E> 0$. Numerical maximization of $E$ leads 
to the optimal measurement basis
\begin{eqnarray}
	 \ket{\alpha_1}&=&(\ket{0}-\E^{-\I\varphi}\ket{1})/\sqrt{2} \nonumber \\
	 \ket{\alpha_2}&=&(\ket{0}+\E^{-\I(\frac{\pi}{2}-\varphi)}\ket{1})/\sqrt{2}
	  \nonumber\\
	 \ket{\beta_1}&=&(\ket{0}-\E^{+\I(\frac{\pi}{2}-\varphi)}\ket{1})/\sqrt{2}
	  \nonumber\\
	 \ket{\beta_2}&=&(\ket{0}+\E^{+\I\varphi}\ket{1})/\sqrt{2}
	  \nonumber\\
	 \ket{\gamma_1}&=&(\ket{0}+\E^{-\I\varphi}\ket{1})/\sqrt{2}
	  \nonumber\\ \ket{\gamma_2}&=&(\ket{0}-\E^{-\I(\frac{\pi}{2}-\varphi)}\ket{1})/\sqrt{2}
	 \end{eqnarray}
with the phase $\varphi\approx 0.138 \pi$ for $\alpha=2$. Here, the two 
measurement directions in a single subsystems are not orthogonal anymore. 
The result of the entanglement detection is shown in \fig{fig:rhoalpha}. 
For $\alpha < 2.4$ our criterion detects entanglement independently 
of whether we use the measurement basis optimized for $\alpha=2$ 
(\textcolor{red}{$+$}) or optimized the measurement basis for each 
$\alpha$ separately (\textcolor{blue}{$\times$}). For $\alpha<2.4$ both 
methods lead to the same quantity $E$ and to the same measurement basis. 
For $2.4 <\alpha$ also the optimized measurement basis does not detect the
states, since it leads to $E=0$. Note that in contrast to $\varrho_{abc}$, 
the entanglement of this state cannot be detected by the criteria developed 
from Ref.\cite{Guehne2010}, but there exist refined criteria which detect
the entanglement in the whole region $2 \leq \alpha \leq 2\sqrt{2}$
\cite{2011PhLA..375..406G}.

\begin{figure}
 \includegraphics[width=0.4\textwidth]{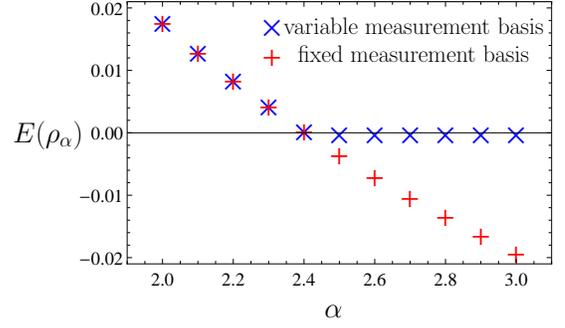}
 \caption{Entanglement detection of $\varrho_\alpha$ \eq{def:rhoalpha} 
 with the help of \eq{eq:Cauchy4}. $E>0$ indicates entanglement. For 
 \textcolor{blue}{$\times$} the measurement basis was optimized for 
 each $\alpha$, whereas for   \textcolor{red}{$+$} we used the 
 optimized basis for $\alpha=2$ for all $\alpha$. See the text 
 for further details. \label{fig:rhoalpha}}
\end{figure}


\section{Connections to existing criteria\label{sec:other_criteria}}

As explained in the introduction, one of the main motivations of our paper
is to present an unified view on several existing entanglement criteria. 
So in this section we show that several other entanglement criteria based 
on inequalities are applications of \eq{eq:Cauchy2}, although they were 
originally proven in a different way. 

\subsection{The criterion of Hillery and Zubairy}
The entanglement criterion
\BE
| \langle   A^\dagger    B\rangle_\sep|^2\leq  \langle   A^\dagger  A   B^\dagger  B\rangle_\sep\label{eq:alg_zubairy1}
\EE
from Ref.~\cite{Hillery2009} is a special case of \eq{eq:Cauchy2} where we 
have set $  A_2$ and $  B_1$ equal to unity and 
$  A_1 =   A^\dagger$ and $  B_2=  B$. By identifying $  A=a^m$ with $a$ 
being the annihilation operator of system A and $  B=(b^\dagger)^n$ with 
$b^\dagger$ being the creation operator of system $B$  we rederive
their original criterion
\BE
|\langle a^m (b^\dagger)^n\rangle_\sep|^2\leq \langle  (a^\dagger)^m a^m  (b^\dagger)^n b^b\rangle_\sep
\EE
given in Ref.~\cite{Hillery2006}. On the other hand, 
\BE
|\langle a^m b^n\rangle_\sep|^2\leq \langle  (a^\dagger)^m a^m \rangle_\sep \langle(b^\dagger)^n b^b\rangle_\sep
\EE
which was also derived in Ref.~\cite{Hillery2006} belongs to the special case  
\BE
| \langle   A  B\rangle_\sep|^2\leq  \langle   A^\dagger  A
\rangle_\sep\langle  B^\dagger  B\rangle_\sep \label{eq:alg_zubairy2}
\EE
where we have set  $  A_1$ and $  B_1$  equal to unity.

Let us compare our entanglement criterion \eq{eq:Cauchy2} with 
the criteria of the type \eq{eq:alg_zubairy2} and \eq{eq:alg_zubairy1} 
derived by Hillery and Zubairy for two qubit systems. By choosing 
\BE
  A_1=\ket{1}\bra{0},\;  A_2=\ket{0}\bra{0},\;  B_1=\ket{1}\bra{0},\;  B_2=\ket{0}\bra{0}\label{eq:operators}
\EE 
we obtain from \eq{eq:Cauchy2} that all separable states $\varrho$ 
obey 
\BE
|\varrho_{00,11}^\sep|^2\leq\varrho_{01,01}^\sep\varrho_{10,10}^\sep\label{eq:seevinck_d2}.
\EE
On the other hand  \eq{eq:alg_zubairy1} and \eq{eq:alg_zubairy2} 
transform for the choice of and $  A=  A_1$ and $  B=  B_1$ into
\begin{eqnarray}
|\varrho_{00,11}^\sep|^2&\leq&\varrho_{01,01}^\sep\\
|\varrho^\sep_{00,11}|^2&\leq&(\varrho^\sep_{10,10}+\varrho^\sep_{11,11})(\varrho^\sep_{01,01}+\varrho^\sep_{11,11}).
\end{eqnarray}
Due to the normalization of the state the coefficients of the matrix 
obey $0 \leq \varrho_{jk,jk}\leq 1$  and therefore these criteria 
are weaker then \eq{eq:Cauchy2}. However, \eq{eq:alg_zubairy1}  
requires only the estimation of a single expectation value or matrix 
entry instead of two for our criterion, and  \eq{eq:alg_zubairy2} only 
requires expectation values depending on single subsystems rather 
than correlations of both systems required for \eq{eq:Cauchy2}.

As suggested in Ref.~\cite{Hillery2006} \eq{eq:alg_zubairy1} can 
be generalized to
\BE
|\langle \prod\limits_{k=1}^{N}  A_k\rangle_\sep|^2 \leq  \langle \prod\limits_{k=1}^j  A_k^\dagger   A_k  \prod\limits_{k=j+1}^{N}   A_k   A_k^\dagger\rangle _\sep\label{eq:Zubairy_gen1}
\EE
with $ A_k$ being an operator acting on system $k$. This inequality holds 
true not only for seperable states, but also for biseperable states with 
respect to the partition ${1,2,\dots,j|j+1,\dots N}$. As a consequence, 
the inequality checks only if the state is biseperable with respect to 
a certain bipartition and does not check for real multipartite entanglement. 

\subsection{The criterion of Hillery et al.}

In Ref.~\cite{Hillery2010} the authors used a similar method to ours 
to generalize  \eq{eq:alg_zubairy2} to multipartite states: (i) First, they
write the separable state as a convex set of pure product states. (ii) Then, they
write every expectation value as the product of expectation values of single 
subsystems and apply the CS inequality to each of them. (iii) Finally, they
use the H\"older inequality \eq{eq:hoelder} to imply convexity. 
In this way they derive the criterion
\BE
 |\langle \prod\limits_{k=1}^n  A_k\rangle_\sep|\leq\prod\limits_{k=1}^n\langle (  A^\dagger_k A_k)^{n/2} \rangle_\sep^{1/n} \label{eq:Hillery1}
 \EE
 where $A_k$ denotes an operator acting on subsystem $k$.
For their second criterion
 \BE
  |\langle \prod\limits_{k=1}^n  A_k\rangle_\sep|\leq\langle\Big( \frac{1}{n}\sum\limits_{k=1}^n   A^\dagger_k A_k\Big)^{n/2} \rangle_\sep\label{eq:Hillery2}
\EE
they used in addition that the geometric mean is smaller or equal 
than the arithmetic mean between step (ii) and (iii). For some states, 
both criteria are equal, for some \eq{eq:Hillery1} is stronger and for 
others \eq{eq:Hillery2} is stronger. However, both criteria used again 
only a single operator per subsystem. Furthermore, after step (ii) 
the property of product states was not used anymore whereas in 
our criteria we used these properties  again to recombine expectation 
values. Therefore, our criterion is able to detect PPT states 
like $\varrho_{abc}$ defined in \eq{def:rhoabc} whereas 
\eq{eq:Hillery1} cannot, which follows from the following 
theorem:

\begin{theorem}
Assume a state $\varrho$ which is biseperable under some 
partitions in the following way: For any pair of subsystems
$k$ and $j$ there there exists a bipartion $M | \bar M$ of the $N$
particles such that $k \in M$ and $j \in \bar M$ and the state
$\vr$ is biseparable for this bipartition $M | \bar M$.
Then the criterion \eq{eq:Hillery1} cannot detect the 
entanglement of this state.
\end{theorem}

Note that by far not all bipartitions have to be separable to fulfill this
condition, the minimal number of different bipartitions needed for this 
assumption scales like $\log_2(N)$. 

\begin{proof}

Without loss of generality we assume there exists a bipartition $A_1,\dots,A_l|A_{l+1},\dots A_n$, and therefore
 \BE
 |\langle\prod\limits_{k=1}^n  A_k\rangle|\leq \sum\limits_{j}p_j\langle \prod\limits_{k=1}^l  A_k^\dagger  A_k\rangle_j^{1/2}\langle \prod\limits_{k=l+1}^n  A_k^\dagger  A_k\rangle_j^{1/2}.
 \EE
With the help of the generalized H\"older inequality \eq{eq:hoelder} 
we obtain
\BE
\begin{split}
 |\langle\prod\limits_{k=1}^n  A_k\rangle|\leq & \Big(\sum\limits_{j}p_j\langle \prod\limits_{k=1}^l  A_k^\dagger  A_k\rangle_j^{n/2l}\Big)^{l/n} \\
 &\times\Big(\sum\limits_{j}p_j\langle \prod\limits_{k=l+1}^n  A_k^\dagger  A_k\rangle_j^{n/2(n-l)}\Big)^{(n-l)/n}.
 \end{split}
\EE
by choosing $r=l/n$ and $s=(n-l)/n$. Again, with the help of \eq{eq:posop} we find
\BE
\begin{split}
|\langle\prod\limits_{k=1}^n  A_k\rangle|\leq & \Big(\langle \prod\limits_{k=1}^l\left(  A_k^\dagger  A_k\right)^{n/2l}\rangle\Big)^{l/n} \\
 &\times\Big(\langle \prod\limits_{k=l+1}^n\left(  A_k^\dagger  A_k\right)^{n/2(n-l)}\rangle\Big)^{(n-l)/n}.
 \end{split}
\EE
For each of the expectation values $\langle\prod   A_k^\dagger   A_k\rangle$ 
exist again a bipartition (with different $p_j$ and $\psi_j$) and therefore 
we arrive at
\BE
\begin{split}
&\Big(\langle \prod\limits_{k=1}^l\left(  A_k^\dagger  A_k\Big)^{n/2l}\rangle\right)^{l/n}  \\ =& \Big(\sum\limits_{j}p_j\langle \prod\limits_{k=1}^x\left(  A_k^\dagger  A_k\right)^{n/2l}\rangle_j \langle \prod\limits_{k=x+1}^l\left(  A_k^\dagger  A_k\right)^{n/2l}\rangle_j\Big)^{l/n}  \\
\leq &\Big(\langle \prod\limits_{k=1}^x\left(  A_k^\dagger  A_k\right)^{n/2x}\rangle\Big)^{x/n}\\ &\times\Big(\langle \prod\limits_{k=x+1}^l\left(  A_k^\dagger  A_k\right)^{n/2(l-x)\rangle}\Big)^{(l-x)/n}
\end{split}
\EE
by using again the generalized H\"older equation, now with $r = x/l$ and
$s= (l-x)/l$ and \eq{eq:posop}. This argument can be applied until all 
expectation values are expectation values of a single subsystem and 
as a consequence  we derive
\BE
 |\langle \prod\limits_{k=1}^n  A_k\rangle|\leq\prod\limits_{k=1}^n\langle (  A^\dagger_k A_k)^{n/2} \rangle^{1/n}.
\EE
This shows that the entanglement of an  $N$-partite state $\varrho$ cannot be 
detected by \eq{eq:Hillery1}, if for every two subsystem $i$ and $j$, there
exist a bipartition of $\varrho$ where $i$ and $j$ can be separated. Furthermore, 
it shows that the criterion \eq{eq:Hillery1} is just a recursive application 
of the criterion \eq{eq:alg_zubairy2} for bipartite entanglement from 
Hillery and Zubairy.  
\end{proof}

   
\subsection{The criteria from Shchukin and Vogel}
The criterion 
\BE
\begin{split}
|\langle a^{\dagger m} a^n a^{\dagger p} a^q  b^{\dagger s} b^r b^{\dagger k} b^l\rangle_\sep|^2 \leq &\langle a^{\dagger m} a^n a^{\dagger n} a^m  b^{\dagger l} b^k b^{\dagger k} b^l\rangle_\sep\\& \times \langle a^{\dagger q} a^p a^{\dagger p} a^q  b^{\dagger s} b^r b^{\dagger r} b^s\rangle_\sep \label{eq:Shchukin}
\end{split}
\EE
derived   by Shchukin and Vogel in \cite{Shchukin2005} was originally 
proved with the help of the PPT criterion 
\cite{Peres1996,Horodecki1996,Horodecki1997} and the matrix of moments. 

Similarly, our criterion \eq{eq:Cauchy2} follows also directly from 
the PPT criterion in the bipartite case. However, our criterion does 
not detect all NPT states. Nevertheless, \eq{eq:Cauchy2} can be used 
to prove the criteria \eq{eq:Shchukin} by identifying 
\begin{eqnarray}
  A_1=\left(a^\dagger\right)^ma^n&,&  A_2=\left(a^\dagger\right)^pa^q,\\
  B_1=\left(b^\dagger\right)^sb^r&,&  B_2=\left(b^\dagger\right)^kb^l.
\end{eqnarray}
As a consequence, all entangled states which can be detected with the 
criterion \eq{eq:Shchukin} derived by Shchukin and Vogel can be also 
detected with our criterion  \eq{eq:Cauchy2}.

\subsection{The criteria by G\"uhne and Seevinck}

In Ref.~\cite{Guehne2010} the investigation of multiparticle 
entanglement was started by deriving criteria for bipartite 
entanglement, 
\begin{eqnarray}
|\varrho^\sep_{1,8}|&\leq&\sqrt{\varrho^\sep_{2,2}\varrho^\sep_{7,7}}\label{eq:guehne1}\\
|\varrho^\sep_{1,8}|&\leq&\sqrt{\varrho^\sep_{3,3}\varrho^\sep_{6,6}}\label{eq:guehne2}\\
|\varrho^\sep_{1,8}|&\leq&\sqrt{\varrho^\sep_{4,4}\varrho^\sep_{5,5}}\label{eq:guehne3}
\end{eqnarray}
with the standard product basis $\{\ket{000},\ket{001},\dots,\ket{111}\}$. 
Although they were proven in another way, they directly follow from our criterion \eq{eq:Cauchy2} by interpreting the systems $A$ and $B$ as two sets of systems 
instead of two single systems.

Using the fact that the geometric mean is smaller or equal to the arithmetic mean, which denotes
\BE
\prod\limits_{j=1}^N a_k^{1/N} \leq \frac{1}{N}\sum\limits_{j=1}^N a_k 
\EE
these three equations lead to the criteria
\begin{eqnarray}
|\varrho^\sep_{1,8}|&\leq&\frac{1}{2}(\varrho^\sep_{2,2}+\varrho^\sep_{7,7})\\
|\varrho^\sep_{1,8}|&\leq&\frac{1}{2}(\varrho^\sep_{3,3}+\varrho^\sep_{6,6})\\
|\varrho^\sep_{1,8}|&\leq&\frac{1}{2}(\varrho^\sep_{4,4}+\varrho^\sep_{5,5})
\end{eqnarray}
used  also in Ref.~\cite{Duer2000}.

By a convex combination of the three equations \eq{eq:guehne1},
\eq{eq:guehne2} and \eq{eq:guehne3} one can derive the condition
\BE
|\varrho^\sep_{1,8}|\leq \sqrt{\varrho^\sep_{2,2}\varrho^\sep_{7,7}}+\sqrt{\varrho^\sep_{3,3}\varrho^\sep_{6,6}}+\sqrt{\varrho^\sep_{4,4}\varrho^\sep_{5,5}}
\EE
which is valid for all biseparable three qubit states, and is only violated 
by fully entangled states. This, of course, can not be derived by our framework, 
since we are only dealing with criteria excluding full separability.

However, by multiplying the three equations \eq{eq:guehne1}, \eq{eq:guehne2} and \eq{eq:guehne3} one finds the inequality 
\BE
|\varrho^\sep_{1,8}|\leq (\varrho^\sep_{2,2}\varrho^\sep_{3,3}\varrho^\sep_{4,4}\varrho^\sep_{5,5}\varrho^\sep_{6,6}\varrho^\sep_{7,7})^{1/6}\label{eq:seefinck_sep}.
\EE
which is valid for all fully seperable states.
This inequality  can be proven also directly by using 
our scheme from \Sec{sec:scheme} to develop the criterion
 \BE
 \begin{split}
   &|\langle   A_1  A_2  B_1  B_2  C_1  C_2\rangle_\sep|\\
   \leq&\sqrt[6]{\langle   A_2^\dagger   A_2  B_2^\dagger  B_2   C_1  C_1^\dagger \rangle_\sep\langle   A_2^\dagger  A_2   B_1  B_1^\dagger   C_2^\dagger   C_2\rangle_\sep} \\ 
  &\times \sqrt[6]{\langle   A_1  A_1^\dagger   B_2^\dagger   B_2  C_1   C_1^\dagger\rangle_\sep\langle   A_1  A_1^\dagger   B_2^\dagger  B_2   C_2^\dagger   C_2\rangle_\sep} \\
  &\times \sqrt[6]{\langle   A_2^\dagger   A_2   B_1  B_1^\dagger   C_1  C_1^\dagger\rangle_\sep\langle   A_1   A_1^\dagger   B_1  B_1^\dagger   C_2^\dagger  C_2\rangle_\sep}\label{eq:Cauchy6},
  \end{split}
 \EE
 and by using the operators defined in \eq{def:op1} and \eq{def:op2}. 
 Since \eq{eq:seefinck_sep} is a product of three criteria for 
 bi-separability, it cannot detect the entanglement of PPT states 
 like $\varrho_{abc}$ \eq{def:rhoabc}. However, as noted in Ref.~\cite{Guehne2010}
 in \eq{eq:seefinck_sep} it is possible to substitute certain matrix entries $\varrho_{j,j}$ by others, for example $\varrho_{2,2}\varrho_{3,3}\rightarrow \varrho_{1,1}\varrho_{4,4}$ to get different criteria. The substitution of matrix 
 entries is equivalent to a different combination of the operators $  A_j$, $  B_j$ 
 and $  C_j$ in our scheme. These new criteria  are able to detect 
 PPT states, for example the criterion
\BE
|\varrho_{1,8}^\sep|\leq (\varrho^\sep_{1,1}(\varrho^\sep_{4,4})^2\varrho^\sep_{5,5}\varrho^\sep_{6,6}\varrho^\sep_{7,7})^{1/6}
\label{eq:guehneabc}
\EE
derived by substituting $\varrho_{2,2}\varrho_{3,3}\rightarrow \varrho_{1,1}\varrho_{4,4}$ in \eq{eq:seefinck_sep}, 
leads to $ 1\leq (abc)^{-1/6}$ for the PPT-state $\varrho_{abc}$ 
given in \eq{def:rhoabc} which is violated for $abc>1$. With a 
second similar equations, the entanglement can be detected for 
$ 1\leq (abc)^{1/6}$. As a consequence, the criteria of 
Ref.~\cite{Guehne2010} detect entanglement if $abc\neq 1$ similar 
to our criterion shown in \Sec{sec:application}. However, 
by translating our criterion from \Sec{sec:application} into 
density matrix representation
\BE
|\varrho^\sep_{1,8}|\leq \sqrt[4]{\varrho^\sep_{1,1}\varrho^\sep_{4,4}\varrho^\sep_{6,6}\varrho^\sep_{7,7}}\label{eq:woelkabc}
\EE
we see that the criterion of Ref.~\cite{Guehne2010} is a weighted mean of 
the criteria \eq{eq:guehne3} and \eq{eq:woelkabc}. Therefore \eq{eq:guehneabc} 
is only equally or less strong than testing \eq{eq:guehne3} and \eq{eq:woelkabc} separately. This was also noted in Ref.~\cite{2011PhLA..375..406G}.



\subsection{The criterion by Huber et~al.}
  
In Ref.~\cite{Huber2010} the criterion
\BE
\begin{split}
&\sqrt{\Re \left[\bra{\Phi_\ps}(\mathbbm{1}_A \otimes \Pi_B)^\dagger \varrho_\sep^{\otimes m}(\Pi_A \otimes \mathbbm{1}_B )\ket{\Phi_\ps}\right]} \\ &\leq \sqrt{\bra{\Phi_\ps}\varrho_\sep^{\otimes m}\ket{\Phi_\ps}}
\end{split}
\EE
where $\varrho_\sep^{\otimes m}$ is an m-fold tensor product of the density 
matrix $\varrho_\sep$,  $\ket{\Phi_\ps}$ being a product state of the m-tupled 
system and $\Pi$ being the cyclic permutation operator. For the case $m=2$ 
this criterion transforms exactly into \eq{eq:optbisep} and is therefore 
equivalent to our criterion if one chooses the operators to be projectors 
on a single state.

For the case $m>2$ the criteria can be led back to a combination of the criterion with $m=2$ and different states $\ket{\Phi_\ps}$. Therefore, the criteria for $m>2$ do not show any advances compared to the criterion for $m=2$.

In the multipartite case, the generalization of the criterion for $m=2$ reads
\BE
\begin{split}
\sqrt{\bra{\Phi_\ps}\varrho_\bs^{\otimes m}\bf{\Pi}\ket{\Phi_\ps}}&\\
-\sum\limits_j\sqrt{\bra{\Phi_\ps}\mathcal{P}^\dagger_j\varrho_\bs^{\otimes m}\mathcal{P}_j\ket{\Phi_\ps}}&\leq 0
\end{split}
\EE
which is valid for all biseparable (bs) states. Here, the operator $\bf{\Pi}$ performs simultaneous permutations on all subsystems and $\mathcal{P}_j$ performs only a permutation on the subsystem $A_j$ for the different bipartitions $j$.

In a similar way, our criteria can be transformed to detect only genuine multipartite entanglement. Let $X_k$ and $Y_k$ be operators acting on subsystem $k$. Furthermore, let $S_j$ be the set of all pure biseparable states which are biseparable under the 
bipartition $A_j|B_j$. As a consequence,  every biseparable state 
\BE
\varrho_\bs\equiv \sum\limits_j \sum\limits_{s_j \in S_j} p_{s_j} \ket{s_j}\bra{s_j}
\EE
can be written as a convex combination of states $s_j$ belonging to different sets $S_j$ with weights $p_{s_j}$. In this case we find the inequality
\BE
\begin{split}
\big|\langle \prod\limits_k X_kY_k\rangle_\bs\big|\leq&
\sum\limits_j \sum\limits_{s_j \in S_j} p_{s_j} \sqrt{\langle \prod\limits_{k\in A_j} X_kX_k^\dagger  \prod\limits_{k\in Bj} Y_k^\dagger Y_k\rangle_{s_j}}\\
&\times \sqrt{\langle \prod\limits_{k\in A_j} Y_k^\dagger Y_k  \prod\limits_{k\in Bj} X_k X_k^\dagger\rangle_{s_j}}.
\end{split}
\EE
Since $X_kX_k^\dagger$ and $Y_k^\dagger Y_k$ are positive operators we 
increase the right side of this inequality be increasing the summation 
range from $S_j$ to $S=\cup S_j$. As a consequence, we obtain the criterion
\BE
\begin{split}
\big|\langle \prod\limits_k X_kY_k\rangle_\bs\big|\leq&
\sum\limits_j  \sqrt{\langle \prod\limits_{k\in A_j} X_kX_k^\dagger  \prod\limits_{k\in Bj} Y_k^\dagger Y_k\rangle_{\bs}}\\
&\times \sqrt{\langle \prod\limits_{k\in A_j} Y_k^\dagger Y_k  \prod\limits_{k\in Bj} X_k X_k^\dagger\rangle_{\bs}}.
\end{split}
\EE
for genuine multipartite entanglement.


\section{Conclusion}

In this paper we have shown how to develop entanglement criteria with 
the help of the Cauchy-Schwarz and the H\"older inequality and the 
properties of separable states. In the bipartite case, our criterion 
is strongly connected to the PPT-criterion and only NPT-states can be 
detected. We demonstrate that all two-qubit states and all NPT-entangled 
states mixed with white noise can be detected with our criteria. 
However, not all NPT-states can be detected with our criterion, 
since it operates in a $2\times2$ dimensional subspace only. 
However, we suggested some ideas for an additional entanglement criterion 
which may detect such states.

We generalized our criteria to the multipartite case with the help 
of the H\"older inequality. We demonstrated that with our scheme 
it is possible to detect the entanglement of states which are 
separable under every bipartite split. As a consequence, our 
criteria are not restricted to NPT entangled states in 
the multipartite case. We also explained how to transform 
our criteria in such a way that they detect only 
genuine multipartite entanglement. 

Furthermore, we showed that some already existing criteria 
for bipartite and multipartite entanglement  which have been 
proven already with other methods, are direct consequences 
of the Cauchy-Schwarz inequality. As a consequence, our 
method to derive entanglement criteria is more fundamental.
We demonstrated that criteria which use two operators 
per subsystems are in general stronger than criteria 
which use only a single operator per subsystem. As 
a consequence, our methods lead to stronger entanglement 
criteria.

\section{Acknowledgement}

S.W. thanks Suhail Zubairy for fruitful discussions. This work has been 
supported by the EU (Marie Curie CIG 293993/ENFOQI and Marie Curie 
IEF 302021/QUACOCOS), the BMBF (Chist-Era Project QUASAR), the 
FQXi Fund (Silicon Valley Community Foundation), and the DFG.

\bibliographystyle{apsrev}
\bibliography{entanglement}

\begin{thebibliography}{21}
\expandafter\ifx\csname natexlab\endcsname\relax\def\natexlab#1{#1}\fi
\expandafter\ifx\csname bibnamefont\endcsname\relax
  \def\bibnamefont#1{#1}\fi
\expandafter\ifx\csname bibfnamefont\endcsname\relax
  \def\bibfnamefont#1{#1}\fi
\expandafter\ifx\csname citenamefont\endcsname\relax
  \def\citenamefont#1{#1}\fi
\expandafter\ifx\csname url\endcsname\relax
  \def\url#1{\texttt{#1}}\fi
\expandafter\ifx\csname urlprefix\endcsname\relax\def\urlprefix{URL }\fi
\providecommand{\bibinfo}[2]{#2}
\providecommand{\eprint}[2][]{\url{#2}}

\bibitem[{\citenamefont{Peres}(1996)}]{Peres1996}
\bibinfo{author}{\bibfnamefont{A.}~\bibnamefont{Peres}},
  \bibinfo{journal}{Phys. Rev. Lett.} \textbf{\bibinfo{volume}{77}},
  \bibinfo{pages}{1413} (\bibinfo{year}{1996}).

\bibitem[{\citenamefont{Horodecki et~al.}(1996)\citenamefont{Horodecki,
  Horodecki, and Horodecki}}]{Horodecki1996}
\bibinfo{author}{\bibfnamefont{P.}~\bibnamefont{Horodecki}},
  \bibinfo{author}{\bibfnamefont{R.}~\bibnamefont{Horodecki}},
  \bibnamefont{and}
  \bibinfo{author}{\bibfnamefont{M.}~\bibnamefont{Horodecki}},
  \bibinfo{journal}{Phys. Lett. A} \textbf{\bibinfo{volume}{223}},
  \bibinfo{pages}{1} (\bibinfo{year}{1996}).

\bibitem[{\citenamefont{Shchukin and Vogel}(2005)}]{Shchukin2005}
\bibinfo{author}{\bibfnamefont{E.}~\bibnamefont{Shchukin}} \bibnamefont{and}
  \bibinfo{author}{\bibfnamefont{W.}~\bibnamefont{Vogel}},
  \bibinfo{journal}{Phys. Rev. Lett.} \textbf{\bibinfo{volume}{95}},
  \bibinfo{pages}{230502} (\bibinfo{year}{2005}).

\bibitem[{\citenamefont{Hillery and Zubairy}(2006)}]{Hillery2006}
\bibinfo{author}{\bibfnamefont{M.}~\bibnamefont{Hillery}} \bibnamefont{and}
  \bibinfo{author}{\bibfnamefont{M.~S.} \bibnamefont{Zubairy}},
  \bibinfo{journal}{Phys. Rev. Lett.} \textbf{\bibinfo{volume}{96}},
  \bibinfo{pages}{050503} (\bibinfo{year}{2006}).

\bibitem[{\citenamefont{G\"uhne and Seevinck}(2010)}]{Guehne2010}
\bibinfo{author}{\bibfnamefont{O.}~\bibnamefont{G\"uhne}} \bibnamefont{and}
  \bibinfo{author}{\bibfnamefont{M.}~\bibnamefont{Seevinck}},
  \bibinfo{journal}{New J. of Phys.} \textbf{\bibinfo{volume}{12}},
  \bibinfo{pages}{053002} (\bibinfo{year}{2010}).

\bibitem[{\citenamefont{Hillery et~al.}(2010)\citenamefont{Hillery, Dung, and
  Zheng}}]{Hillery2010}
\bibinfo{author}{\bibfnamefont{M.}~\bibnamefont{Hillery}},
  \bibinfo{author}{\bibfnamefont{H.~T.} \bibnamefont{Dung}}, \bibnamefont{and}
  \bibinfo{author}{\bibfnamefont{H.}~\bibnamefont{Zheng}},
  \bibinfo{journal}{Phys. Rev. A} \textbf{\bibinfo{volume}{81}},
  \bibinfo{pages}{062322} (\bibinfo{year}{2010}).

\bibitem[{\citenamefont{D\"ur and Cirac}(2000)}]{Duer2000}
\bibinfo{author}{\bibfnamefont{W.}~\bibnamefont{D\"ur}} \bibnamefont{and}
  \bibinfo{author}{\bibfnamefont{J.~I.} \bibnamefont{Cirac}},
  \bibinfo{journal}{Phys. Rev. A.} \textbf{\bibinfo{volume}{61}},
  \bibinfo{pages}{042314} (\bibinfo{year}{2000}).

\bibitem[{\citenamefont{Huber et~al.}(2010)\citenamefont{Huber, Mintert,
  Gabriel, and Hiesmayr}}]{Huber2010}
\bibinfo{author}{\bibfnamefont{M.}~\bibnamefont{Huber}},
  \bibinfo{author}{\bibfnamefont{F.}~\bibnamefont{Mintert}},
  \bibinfo{author}{\bibfnamefont{A.}~\bibnamefont{Gabriel}}, \bibnamefont{and}
  \bibinfo{author}{\bibfnamefont{B.~C.} \bibnamefont{Hiesmayr}},
  \bibinfo{journal}{Phys. Rev. Lett.} \textbf{\bibinfo{volume}{104}},
  \bibinfo{pages}{210501} (\bibinfo{year}{2010}).

\bibitem[{\citenamefont{Cauchy}(1821)}]{Cauchy}
\bibinfo{author}{\bibfnamefont{A.-L.} \bibnamefont{Cauchy}},
  \emph{\bibinfo{title}{Analyse alg\'ebrique}} (\bibinfo{publisher}{Imprimerie
  Royale, Paris}, \bibinfo{year}{1821}).

\bibitem[{\citenamefont{Hillery et~al.}(2009)\citenamefont{Hillery, Dung, and
  Niset}}]{Hillery2009}
\bibinfo{author}{\bibfnamefont{M.}~\bibnamefont{Hillery}},
  \bibinfo{author}{\bibfnamefont{H.~T.} \bibnamefont{Dung}}, \bibnamefont{and}
  \bibinfo{author}{\bibfnamefont{J.}~\bibnamefont{Niset}},
  \bibinfo{journal}{Phys. Rev. A} \textbf{\bibinfo{volume}{80}},
  \bibinfo{pages}{052335} (\bibinfo{year}{2009}).

\bibitem[{\citenamefont{Johnston}(2013)}]{Johnston2013}
\bibinfo{author}{\bibfnamefont{N.}~\bibnamefont{Johnston}},
  \bibinfo{journal}{Phys. Rev. A} \textbf{\bibinfo{volume}{87}},
  \bibinfo{pages}{064302} (\bibinfo{year}{2013}).

\bibitem[{\citenamefont{Rana}(2013)}]{Rana2013}
\bibinfo{author}{\bibfnamefont{S.}~\bibnamefont{Rana}}, \bibinfo{journal}{Phys.
  Rev. A} \textbf{\bibinfo{volume}{87}}, \bibinfo{pages}{054301}
  (\bibinfo{year}{2013}).

\bibitem[{\citenamefont{Sanpera et~al.}(1998)\citenamefont{Sanpera, Tarrach,
  and Vidal}}]{Sanpera1998}
\bibinfo{author}{\bibfnamefont{A.}~\bibnamefont{Sanpera}},
  \bibinfo{author}{\bibfnamefont{R.}~\bibnamefont{Tarrach}}, \bibnamefont{and}
  \bibinfo{author}{\bibfnamefont{G.}~\bibnamefont{Vidal}},
  \bibinfo{journal}{Phys. Rev. A} \textbf{\bibinfo{volume}{58}},
  \bibinfo{pages}{826} (\bibinfo{year}{1998}).

\bibitem[{\citenamefont{Augusiak et~al.}(2008)\citenamefont{Augusiak,
  Demianowicz, and Horodecki}}]{Horodecki2008}
\bibinfo{author}{\bibfnamefont{R.}~\bibnamefont{Augusiak}},
  \bibinfo{author}{\bibfnamefont{M.}~\bibnamefont{Demianowicz}},
  \bibnamefont{and}
  \bibinfo{author}{\bibfnamefont{P.}~\bibnamefont{Horodecki}},
  \bibinfo{journal}{Phys. Rev. A} \textbf{\bibinfo{volume}{77}},
  \bibinfo{pages}{030301} (\bibinfo{year}{2008}).

\bibitem[{\citenamefont{Horodecki et~al.}(1998)\citenamefont{Horodecki,
  Horodecki, and Horodecki}}]{horodecki98a}
\bibinfo{author}{\bibfnamefont{M.}~\bibnamefont{Horodecki}},
  \bibinfo{author}{\bibfnamefont{P.}~\bibnamefont{Horodecki}},
  \bibnamefont{and}
  \bibinfo{author}{\bibfnamefont{R.}~\bibnamefont{Horodecki}},
  \bibinfo{journal}{Phys. Rev. Lett.} \textbf{\bibinfo{volume}{80}},
  \bibinfo{pages}{5239} (\bibinfo{year}{1998}).

\bibitem[{\citenamefont{Werner}(1989)}]{Werner1989}
\bibinfo{author}{\bibfnamefont{R.~F.} \bibnamefont{Werner}},
  \bibinfo{journal}{Phys. Rev. A} \textbf{\bibinfo{volume}{40}},
  \bibinfo{pages}{4277} (\bibinfo{year}{1989}).

\bibitem[{\citenamefont{Hardy et~al.}(1934)\citenamefont{Hardy, Littlewood, and
  P\'{o}lya}}]{Hardy34}
\bibinfo{author}{\bibfnamefont{G.}~\bibnamefont{Hardy}},
  \bibinfo{author}{\bibfnamefont{J.~E.} \bibnamefont{Littlewood}},
  \bibnamefont{and}
  \bibinfo{author}{\bibfnamefont{G.}~\bibnamefont{P\'{o}lya}},
  \emph{\bibinfo{title}{Inequalities}} (\bibinfo{publisher}{Cambridge Univ.
  Press, London}, \bibinfo{year}{1934}).

\bibitem[{\citenamefont{Ac\'{\i}n et~al.}(2001)\citenamefont{Ac\'{\i}n, Bru\ss,
  Lewenstein, and Sanpera}}]{Acin2001}
\bibinfo{author}{\bibfnamefont{A.}~\bibnamefont{Ac\'{\i}n}},
  \bibinfo{author}{\bibfnamefont{D.}~\bibnamefont{Bru\ss}},
  \bibinfo{author}{\bibfnamefont{M.}~\bibnamefont{Lewenstein}},
  \bibnamefont{and} \bibinfo{author}{\bibfnamefont{A.}~\bibnamefont{Sanpera}},
  \bibinfo{journal}{Phys. Rev. Lett.} \textbf{\bibinfo{volume}{87}},
  \bibinfo{pages}{040401} (\bibinfo{year}{2001}).

\bibitem[{\citenamefont{Kay}(2011)}]{Kay2011}
\bibinfo{author}{\bibfnamefont{A.}~\bibnamefont{Kay}}, \bibinfo{journal}{Phys.
  Rev. A} \textbf{\bibinfo{volume}{83}}, \bibinfo{pages}{020303}
  (\bibinfo{year}{2011}).

\bibitem[{\citenamefont{{G{\"u}hne}}(2011)}]{2011PhLA..375..406G}
\bibinfo{author}{\bibfnamefont{O.}~\bibnamefont{{G{\"u}hne}}},
  \bibinfo{journal}{Phys. Lett. A} \textbf{\bibinfo{volume}{375}},
  \bibinfo{pages}{406} (\bibinfo{year}{2011}).

\bibitem[{\citenamefont{Horodecki et~al.}(1997)\citenamefont{Horodecki,
  Horodecki, and Horodecki}}]{Horodecki1997}
\bibinfo{author}{\bibfnamefont{P.}~\bibnamefont{Horodecki}},
  \bibinfo{author}{\bibfnamefont{R.}~\bibnamefont{Horodecki}},
  \bibnamefont{and}
  \bibinfo{author}{\bibfnamefont{M.}~\bibnamefont{Horodecki}},
  \bibinfo{journal}{Phys. Lett. A} \textbf{\bibinfo{volume}{232}},
  \bibinfo{pages}{333} (\bibinfo{year}{1997}).

\end{thebibliography}
\end{document}